\documentclass[letterpaper, 10 pt, conference]{ieeeconf}  
\IEEEoverridecommandlockouts                              
\let\labelindent\relax

\usepackage{booktabs}
\usepackage{multirow}

\usepackage{cite}

\usepackage[utf8]{inputenc}
\usepackage{graphics} 
\usepackage{amsmath} 
\usepackage{amssymb}  
\usepackage{amsfonts}
\usepackage{comment}
\usepackage{enumitem}
\usepackage{cases,xspace,enumitem}

\usepackage{version}
\includeversion{tac}
\excludeversion{arxiv}

\usepackage{tikz}
\usetikzlibrary{positioning}

\usepackage[prependcaption,colorinlistoftodos]{todonotes}

\definecolor{gnred}{RGB}{255,91,89}

\definecolor{gnred1}{RGB}{71,0,0} 
\definecolor{gnred2}{RGB}{117,0,0} 
\definecolor{gnred3}{RGB}{164,0,0} 
\definecolor{gnred4}{RGB}{211,0,0} 
\definecolor{gnred5}{RGB}{255,0,0} 
\definecolor{gnred6}{RGB}{255,42,34} 
\definecolor{gnred7}{RGB}{255,91,89} 

\definecolor{gnblue1}{RGB}{0,36,71}   
\definecolor{gnblue2}{RGB}{0,60,118}  
\definecolor{gnblue3}{RGB}{0,85,164}
\definecolor{gnblue4}{RGB}{0,108,212}
\definecolor{gnblue4}{RGB}{0,108,212}
\definecolor{gnblue5}{RGB}{0,133,255}  
\definecolor{gnblue6}{RGB}{35,156,255} 
\definecolor{gnblue7}{RGB}{88,177,255} 

\definecolor{gnbrown1}{RGB}{71,27,0}  
\definecolor{gnbrown2}{RGB}{117,45,0} 
\definecolor{gnbrown3}{RGB}{164,62,0} 
\definecolor{gnbrown4}{RGB}{211,80,0} 
\definecolor{gnbrown5}{RGB}{255,97,0} 
\definecolor{gnbrown6}{RGB}{255,127,26} 
\definecolor{gnbrown7}{RGB}{255,155,86} 

\renewcommand{\theenumi}{(\roman{enumi})}

\newcommand\Item[1][]{%
  \ifx\relax#1\relax  \item \else \item[#1] \fi
  \abovedisplayskip=0pt\abovedisplayshortskip=0pt~\vspace*{-\baselineskip}}

\usepackage{hyperref}
\hypersetup{
    colorlinks=true,
    linkcolor=blue,
    filecolor=magenta,      
    urlcolor=cyan,
    pdftitle={},
    pdfpagemode=FullScreen,
    }

\newcommand{\e}{\mathrm{e}}

\usepackage{amsthm}

\newtheoremstyle{ieeeconf}
  {0pt}   
  {0pt}   
  {\normalfont}  
  {\parindent}       
  {\itshape} 
  {:}         
  { } 
  {\thmname{#1} \thmnumber{#2}\thmnote{ (#3)}} 
\makeatletter
\renewenvironment{proof}[1][\proofname]{\par
  \pushQED{\qed}%
  \normalfont \topsep\z@
  \trivlist
  \item[\hskip2em
        \itshape
    #1\@addpunct{:}]\ignorespaces
}{%
  \popQED\endtrivlist\@endpefalse
}
\makeatletter

\theoremstyle{ieeeconf}

\newtheorem{defn}{Definition}
\newtheorem{thm}{Theorem}
\newtheorem{lemma}[thm]{Lemma}
\newtheorem{pro}[thm]{Proposition}
\newtheorem{cor}[thm]{Corollary}
\newtheorem{remark}[thm]{Remark}

\newcommand{\R}{\mathbb{R}}
\newcommand{\real}{\mathbb{R}}

\newcommand{\realnonnegative}{\mathbb{R}_{\geq 0}}

\newcommand{\norm}[2]{\|#1\|_{#2}}

\newcommand{\diag}[1]{\operatorname{diag}(#1)}

\newcommand{\map}[3]{#1 \colon #2 \rightarrow #3}

\newcommand{\Lip}{\operatorname{\mathsf{Lip}}}

\newcommand{\mcM}{\mathcal{M}}
\newcommand{\mcN}{\mathcal{N}}

\newcommand{\mcT}{\mathcal{T}}

\newcommand{\mcW}{\mathcal{W}}
\newcommand{\mcX}{\mathcal{X}}
\newcommand{\mcY}{\mathcal{Y}}

\newcommand{\xstar}{x^{\star}}
\newcommand{\ustar}{u^{\star}}

\newcommand{\0}{\mbox{\fontencoding{U}\fontfamily{bbold}\selectfont0}}

\newcommand{\fzero}{\0}

\title{Contractivity of Neural Networks:  \\
Classification, Integral Control and Learning}

\title{Contracting  Neural Networks: \\ Sharp LMI Characterizations for Integral Control and Deep Learning}

\title{Contracting  Neural Networks:  Sharp LMI \textcolor{black}{Conditions} \\with Applications to Integral Control and Deep Learning}

\author{Anand Gokhale, Anton V. Proskurnikov, Yu Kawano, Francesco Bullo \thanks{ This work was in part supported by AFOSR project FA9550-22-1-0059 and by JST FOREST Program Grant Number JPMJFR222E. Anand Gokhale and Francesco Bullo are with
    the Center for Control, Dynamical Systems, and Computation, UC Santa Barbara, Santa Barbara, CA 93106 USA. email:  {\tt\small \{anand\_gokhale,bullo\}@ucsb.edu}. \newline
    Anton V. Proskurnikov is with Department of Electronics and Telecommunications, Politecnico di Torino, Italy, 10129.  email: {\tt\small anton.p.1982@ieee.org} \newline
    Yu Kawano is with the Graduate School of Advanced Science and Engineering, Hiroshima University, Higashi-Hiroshima 739-8527, Japan. email: {\tt\small ykawano@hiroshima-u.ac.jp}  \newline
    Anand Gokhale thanks Alexander Davydov for fruitful discussions.
    }}

\begin{document}

\maketitle

\thispagestyle{empty}
\pagestyle{empty}

\begin{abstract}
This paper studies \textcolor{black}{contractivity of}  
firing-rate and Hopfield recurrent neural networks.   
\textcolor{black}{We derive sharp LMI conditions on the synaptic matrices that characterize contractivity of both architectures, for activation functions} that are either non-expansive or monotone non-expansive, in both continuous and discrete time. We establish structural relationships among these conditions, including connections to Schur diagonal stability \textcolor{black}{and the recovery of optimal contraction rates for symmetric synaptic matrices.}
We demonstrate the utility of these results through two applications. First, we \textcolor{black}{develop an} 
LMI-based design procedure for low-gain integral controllers enabling reference tracking in contracting firing rate networks. Second, we provide an exact parameterization of 
\textcolor{black}{weight matrices that guarantee contraction and use it}
to improve the expressivity of Implicit Neural Networks, achieving competitive performance on image classification benchmarks with fewer parameters.
\end{abstract}
\section{Introduction}
\paragraph{Motivation}
Recurrent neural networks (RNNs) are a foundational architecture in machine learning. Unlike heavily parameterized models such as transformers, RNNs maintain 
\textcolor{black}{low computational and memory requirements}, 
making them well suited for deployment in resource-constrained edge applications and real-time data-driven control~\cite{WDA-ALB-MF:24, WDA-ALB-MF:25}. Furthermore, their structure naturally mirrors biological neural circuits, establishing them as a biologically plausible model in computational neuroscience~\cite{CJR-DHJ-RGB-BAO:08}. Beyond traditional sequence processing, RNNs have recently \textcolor{black}{become central} 
for implicit deep learning; in architectures such as Deep Equilibrium Models (DEQs)~\cite{EW-JZK:20}, 
\textcolor{black}{a cascade of layers is}
replaced by the steady-state equilibrium of an RNN, directly linking model evaluation to the network's asymptotic dynamics.
Continuous-time versions of RNNs are particularly well positioned to exploit these advantages. In hardware implementations, continuous-time dynamics offer 
\textcolor{black}{faster processing and lower power consumption, motivating long-standing}
interest in analog circuit design~\cite{DWT-JJH:84}. 

Across \textcolor{black}{these} domains, reliable implementation requires establishing stability and robustness guarantees to prevent divergent behaviors and ensure safe operation \textcolor{black}{in the presence of noise and model uncertainties.} 
To this end, we employ contraction theory~\cite{FB:26-CTDS} as the primary mathematical framework for certifying robust stability of RNNs.

A key challenge in applying contraction theory to neural networks is the identification of sharp, computationally tractable conditions that guarantee contraction. Recent works~\cite{LEG-FG-BT-AA-AYT:21, WDA-ALB-MF:24} derive stability conditions for globally non-expansive activation functions. However, the most widely used activation functions -- such as ReLU, tanh, and sigmoid -- \textcolor{black}{are not merely non-expansive but also monotone non-decreasing. Neglecting this additional structure may result in}  overly conservative stability guarantees. 
Accordingly, we seek to characterize the maximal set of synaptic weight matrices that render RNNs contracting for such practically important classes of nonlinearities. 

These sharp contractivity guarantees unlock advanced capabilities in both control design and machine learning. In control design, contractivity of the RNN 
\textcolor{black}{allows it to be treated as a fast subsystem.}
Motivated by~\cite{JWSP:21}, we use singular perturbation arguments to develop an LMI-based design strategy for low-gain integral controllers. In machine learning, enlarging the \textcolor{black}{admissible space of weight matrices}
directly improves model expressivity. Furthermore, recent work~\cite{JL-LD-SO-WY:25} shows that the representational capacity of implicit models is visibly enhanced by allowing the fixed point of the implicit layer to be locally, rather than globally, Lipschitz in its input. 
\textcolor{black}{We derive an explicit parameterization of the matrices satisfying our sharp contractivity conditions, which enables
the design of highly expressive implicit networks that are contracting yet are locally Lipschitz in the input.}

\paragraph{Relevant Literature}

Contraction theory~\cite{FB:26-CTDS} provides a powerful framework \textcolor{black}{for the control of nonlinear systems and optimization}, guaranteeing exponential convergence, periodic entrainment, and robustness to noise. 
A contraction-based analysis of continuous time neural networks in the $\ell_1$ and $\ell_\infty$ norms has been presented in~\cite{AD-SJ-FB:20o}, and the sharpest known conditions in the Euclidean norm are limited to the symmetric case~\cite{VC-AG-AD-GR-FB:23c}. A parameterization of discrete time robust RNNs is presented in~\cite{MR-RW-IRM:21}.

Control design for discrete time RNNs has been previously studied in the context of asymptotic stability~\cite{AN-MSH-RDB:22} and incremental ISS properties~\cite{WDA-ALB-MF:24}. Structurally, RNNs can be considered a subclass of Lur'e-type systems, i.e., feedback interconnections of LTI blocks and static nonlinear maps. However, unlike standard Lur'e-type models, RNNs involve external inputs or outputs that may enter or exit the system through a nonlinearity. Lur'e-type systems where the inputs and outputs are linear have been studied in~\cite{VA-ST:19, MG-VA-ST-DA:23} under various assumptions, with~\cite{MG-VA-ST-DA:23} studying integral control of such systems. Here, we study the nonlinear case, and specifically design a reference tracking controller using tools from singular perturbation theory~\cite{JWSP:21}.

\paragraph{Contributions}

This paper studies stability conditions for \textcolor{black}{continuous-time and discrete-time}
firing-rate neural networks (FRNNs) and Hopfield neural networks (HNNs) \textcolor{black}{with non-expansive or monotone non-expansive activation functions,
using the framework of contraction theory}.

First, we establish general conditions for the absolute contractivity of Lur'e systems in Lemma~\ref{lem:absolute_contractivity_lure}. Utilizing this general result, we derive novel sharp LMI conditions that guarantee the contractivity of FRNNs and HNNs across popular classes of activation functions, summarizing these comprehensively in Table~\ref{tab:lmi_conditions}.

Second, we perform a rigorous structural analysis of the conditions presented in Table~\ref{tab:lmi_conditions}. In Theorem~\ref{thm:structural_relationships}, we establish that the set of weight matrices guaranteeing contraction in discrete time is a  subset of those guaranteeing contraction in continuous time. Furthermore, while the literature often assumes component-wise non-expansive nonlinearities~\cite{LEG-FG-BT-AA-AYT:21}, popular activation functions like ReLU, tanh, and sigmoid are also non-decreasing. We show that exploiting this additional constraint yields a strictly more expressive set of allowable weight matrices. Finally, we establish the explicit relationship between our FRNN and HNN conditions. Additionally we demonstrate the tightness of our inequalities by reducing them to Schur diagonal stability~\cite{WDA-ALB-MF:24} in discrete time, recover state-of-the-art guarantees for symmetric weights~\cite{VC-AG-AD-GR-FB:23c}, and connect our results to Lyapunov Diagonal Stability.
 
Third, we demonstrate the practical utility of these conditions through two distinct applications. In the realm of control design, we provide a novel LMI-based procedure that enables robust reference tracking in contracting FRNNs via low-gain integral control; a result tantamount to the DC gain condition for linear systems. We validate this approach numerically on a learned two-tank system. In the context of deep learning, we derive an exact algebraic parameterization of the complete set of weight matrices satisfying our contractivity conditions. We apply this parameterization to improve the expressivity of DEQs; specifically, we utilize input-dependent networks to predict the free variables of our parameterization. This safely enhances the expressivity of the implicit layer while remaining mathematically guaranteed to be contracting, demonstrating strong performance on MNIST and CIFAR-10 image classification tasks.

The rest of this paper is organized as follows: We introduce some preliminaries in Section~\ref{sec:background}, and then provide a general result on Lur'e systems in Section~\ref{sec:lure}. We provide an analysis of the conditions for contractivity of FRNNs and HNNs in Section~\ref{sec:analysis_frnn_hnn}. We study both applications in Section~\ref{sec:applications}.

\section{Background}\label{sec:background}

\subsection{Notation}
We let $\0_{n \times m}$ be an $n \times m$ all zero matrix, $I_n$ be the $n \times n$ identity matrix. For symmetric $A, B \in \R^{n \times n}, A \preceq B$ means that $B - A$ is positive definite. We let $\norm{\cdot}{}$ be a norm on $\R^n$ and its corresponding matrix norm on $\R^{n \times n}$. Given a symmetric positive definite matrix $P \in \R^{n\times n}$, we let $\norm{\cdot}{P}$ be the $P-$weighted Euclidean norm $\norm{x}{P}:= \sqrt{x^\top P x}\;, x \in \R^n$. Given two normed spaces $\mcX, \mcY$, a map $\map{F}{\mcX}{\mcY}$ is said to be Lipschitz from $(\mcX, \norm{\cdot}{\mcX})$ to  $(\mcY, \norm{\cdot}{\mcY})$  with constant $\rho \geq 0$, if for all $x_1, x_2 \in \mcX$, $\norm{F(x_1) - F(x_2)}{\mcY} \leq \rho \norm{x_1 - x_2}{\mcX}$. We let $\diag{A_1, A_2, \cdots, A_n}$ denote a block diagonal matrix when $A_i$ is a square matrix, and a diagonal matrix when $A_i$ is scalar.

\subsection{Contraction theory}

Let the vector field $\map{F}{\realnonnegative \times \R^n}{\R^n}$ be continuous in its first argument and Lipschitz in its second, with $(t,x)\to F(t,x)$. We say that $F$ is strongly infinitesimally contracting with respect to $\norm{\cdot}{P}$  with  rate $c > 0$ if for all $x_1, x_2 \in \R^n$, and $t \geq 0$,
\begin{align}
    (F(t, x_1) - F(t, x_2))^\top P (x_1 - x_2) &\leq -c\norm{x_1 - x_2}{P}^2.
    \label{eq:contractivity_condition} 
\end{align}
If $x(t)$ and $y(t)$ are two trajectories of the system $\dot{x} = F(t,x)$, then $\norm{x(t) - y(t)}{P} \leq \e^{-c(t- t_0)}\norm{x(t_0) - y(t_0)}{P}$, for all $t \geq t_0 \geq 0$. In discrete time, the system $x^+ = F(t,x)$ is strongly contracting with constant $\rho$ if $\Lip_x(F) \leq \rho < 1$. We refer to~\cite{FB:26-CTDS} for a recent review of contraction theory.

\subsection{Matrix inequalities for system analysis}

For our analysis, we utilize incremental matrix multipliers constraints (IMMs)~\cite{MF-AR-HH-MM-GJP:19} paired with the S-lemma.

\begin{defn}[Incremental multiplier matrix]
    Consider a function $\map{\Psi}{\R^n}{\R^m}$, and a symmetric matrix $M \in \R^{(n+m)\times (n+m)}$. The function $\Psi$ is said to admit the incremental multiplier matrix $M$, if, for any $x,y \in \R^n$, 
    \begin{align}
    \begin{bmatrix}
        x - y \\
        \Psi(x) - \Psi(y)
    \end{bmatrix}^\top
    M
    \begin{bmatrix}
        x - y \\
        \Psi(x) - \Psi(y)
    \end{bmatrix}
    \ge 0. \label{eq:imm-d}
\end{align}
\end{defn}
\begin{lemma}\label{lem:IMM_conditions}
    Let $Q \in \R^{n \times n}$ be a 
    positive diagonal matrix. If an elementwise nonlinearity $\map{\Psi}{\R^n}{\R^n}$ is slope-restricted in $[k_1, k_2]$ (where $k_1 \leq k_2$), it admits the multiplier matrix:
    \begin{align*}
        M = \begin{bmatrix}
            -2k_1k_2 Q & (k_1 + k_2) Q \\ 
            (k_1 + k_2) Q  & -2Q
        \end{bmatrix}.
    \end{align*}
\end{lemma}

In the context of neural networks, we define two classes of elementwise nonlinearities based on slope restrictions.

\begin{defn}[CONE and MONE nonlinearities]\label{def:cone_mone}
    An elementwise nonlinearity $\map{\Psi}{\R^n}{\R^n}$ is said to be:
    \begin{enumerate}
        \item {component-wise non-expansive (CONE)} if it is slope-restricted in $[-1, 1]$, and
        \item {monotonically non-decreasing and non-expansive (MONE)} if it is slope-restricted in $[0, 1]$.
    \end{enumerate}
\end{defn}

By applying Lemma~\ref{lem:IMM_conditions} to the bounds defined in Definition~\ref{def:cone_mone}, we obtain the following standard multipliers~\cite[E3.25]{FB:26-CTDS}.

\begin{cor}\label{cor:CONE_MONE_IQC}
  For any positive diagonal matrix $Q \in \R^{n \times n}$, the matrices
  \begin{equation*}
    M_{\textup{CONE}} = \begin{bmatrix}
      Q & \fzero_{n \times n} \\
      \fzero_{n \times n} & -Q
    \end{bmatrix} \enspace\text{and}\enspace
    M_{\textup{MONE}} = \begin{bmatrix}
      \fzero_{n \times n} & Q \\
      Q & -2Q
    \end{bmatrix}
  \end{equation*}
  are incremental multiplier matrices for CONE and MONE
  nonlinearities, respectively.
\end{cor}

\section{Contractivity of Lur'e systems} \label{sec:lure}
In our work, we consider both the continuous and discrete time Lur'e systems. These are presented below, 
\begin{align}
    \dot x &= Ax + B\Psi(Cx), \quad \textrm{and } \label{eq:cts_lure} \\
    x^+ &= Ax + B\Psi(Cx) \label{eq:disc_lure}
\end{align}
Here $x \in \R^n$ is the state, matrices $A \in \R^{n \times n}$, $B \in \R^{n \times m}$, $C \in \R^{m \times n}$ and $\map{\Psi}{\R^m}{\R^m}$ is a nonlinearity. We begin by presenting a general result on the absolute contractivity of Lur'e systems. The term absolute contractivity is similar to the concept of absolute stability~\cite{HKK:02}, and is the property that a  system is contracting for each nonlinearity $\Psi$ obeying a certain constraint.  The continuous time version of the following lemma is presented in~\cite{LDA-MC:13}, we extend it to the discrete time case. 

\begin{lemma}[Absolute contractivity of Lur'e systems] \label{lem:absolute_contractivity_lure}
    Consider the continuous time Lur'e system~\eqref{eq:cts_lure} and let $P \in \R^{n\times n}$ be positive definite. Let $\Psi$ admit an incremental matrix multiplier $M \in \R^{2m \times 2m}$.  Let $\Gamma$ be a block diagonal matrix, whose diagonal entries are $C$ and $I_m$. The system~\eqref{eq:cts_lure} is strongly infinitesimally contracting with respect to $\norm{\cdot}{P}$ with rate $c > 0$ if there exists a $\lambda \geq 0$ such that
    \begin{equation}\label{eq:cts_lure_lmi}
      \begin{bmatrix}
        PA+A^\top P + 2c P & PB \\    B^\top P & \0_{m\times{m}}
      \end{bmatrix} + \lambda
      \Gamma^\top M\Gamma \preceq 0.
    \end{equation}    
    Further, the discrete time system~\eqref{eq:disc_lure} is strongly contracting with respect to $\norm{\cdot}{P}$ with a constant $\rho \in [0, 1)$ if there exists a $\lambda \geq 0$ such that 
    \begin{equation}\label{eq:disc_lure_lmi}
      \begin{bmatrix}
        A^\top P A - \rho^2 P &  A^\top P B  \\
        B^\top P A &  B^\top P B
      \end{bmatrix} + \lambda
      \Gamma^\top M\Gamma \preceq 0.
    \end{equation}  
\end{lemma}
\begin{proof}
  Given ${x_1,x_2\in\R^n}$, adopt the shorthands ${\Delta x =
  x_1-x_2\in\real^n}$, ${\Delta y = y_1-y_2 = C\Delta x\in\real^m}$, and ${\Delta
  \Psi = \Psi(y_1)-\Psi(y_2)\in\real^m}$. Since
  \begin{equation*}
    \begin{bmatrix} \Delta y \\ \Delta \Psi \end{bmatrix}
    =
    \begin{bmatrix} C \Delta x  \\ \Delta \Psi \end{bmatrix}
    =
    \begin{bmatrix} C & \0_{m\times{m}} \\ \0_{m\times{n}} & I_m   \end{bmatrix} 
    \begin{bmatrix} \Delta x \\ \Delta \Psi \end{bmatrix} = \Gamma \begin{bmatrix} \Delta x \\ \Delta \Psi \end{bmatrix},
  \end{equation*}
  we rewrite the incremental multiplier matrix condition~\eqref{eq:imm-d}
  as
  \begin{align} \label{eq:imm-x}
    \begin{bmatrix} \Delta y \\ \Delta \Psi \end{bmatrix}^\top
    M 
    \begin{bmatrix} \Delta y \\ \Delta \Psi \end{bmatrix}
    \geq 0
    \iff
    \begin{bmatrix} \Delta x \\ \Delta \Psi \end{bmatrix}^\top
    \Gamma^\top 
    M
    \Gamma
    \begin{bmatrix} \Delta x \\ \Delta \Psi \end{bmatrix}
    \geq 0.
  \end{align}

  The strong infinitesimal contractivity~\eqref{eq:contractivity_condition} condition for the system~\eqref{eq:cts_lure} may be rewritten as
  \begin{align} 
    &\qquad \qquad\bigl(A \Delta x+B\Delta \Psi\bigr)^\top P \Delta x \leq
    - \eta \Delta x^\top P \Delta x
    \nonumber \\
    & \iff
    \Delta x^\top (A^\top P + PA + 2\eta P)\Delta x + 2\Delta x^\top PB
    \Delta \Psi \leq 0
    \nonumber \\
    & \iff
        \begin{bmatrix} \Delta x \\ \Delta \Psi \end{bmatrix}^\top
    \begin{bmatrix}
          A^\top P + PA + 2\eta P & PB \\
          B^\top P &  \0_{m\times{m}}
    \end{bmatrix}
    \begin{bmatrix} \Delta x \\ \Delta \Psi \end{bmatrix}
    \leq 0.
     \label{eq:quad-ineq-contraction}
  \end{align}
  By invoking S-lemma, the inequality~\eqref{eq:imm-x} implies the inequality~\eqref{eq:quad-ineq-contraction} if the inequality~\eqref{eq:cts_lure_lmi} holds. 
  
  Similarly, in the discrete time case, one may rewrite the contractivity condition using the definition of the Lipschitz bound in $\norm{\cdot}{P}$.
  \begin{align*} 
    &\qquad \qquad\norm{A\Delta x + B \Delta \Psi}{P}^2 \leq \rho^2 \norm{\Delta x}{P}^2\\
    & \iff 
    \begin{bmatrix} \Delta x \\ \Delta \Psi \end{bmatrix}^\top
    \begin{bmatrix}
      A^\top P A - \rho^2 P &  A^\top P B \\
      B^\top P A  &  B^\top P B
    \end{bmatrix}
    \begin{bmatrix} \Delta x \\ \Delta \Psi \end{bmatrix}
    \leq 0.
  \end{align*}
  The final step is again an application of the S-Lemma.
\end{proof}
\section{Contraction of firing rate and Hopfield Neural networks} \label{sec:analysis_frnn_hnn}
We investigate the firing rate neural network (FRNN), and the Hopfield neural network (HNN) in both continuous and discrete time. 
The continuous time dynamics are given below. 
\begin{align}
    \dot{x} &= -x + \Psi(Wx + Bu); \quad y = Cx \label{eq:cts_firing_rate} \\
    \dot{x} &= -x + W\Psi(x) + Bu; \quad y = C\Psi(x). \label{eq:cts_hopfield}
\end{align}
Here, the state $x\in \R^n$, the bias $u \in \R^m, y \in \R^p$, and the synaptic matrix $W \in \R^{n\times n}$, $\Psi(\cdot)$ is an element-wise, slope restricted nonlinearity. Equation~\eqref{eq:cts_firing_rate} represents the firing rate dynamics, and Equation~\eqref{eq:cts_hopfield} represents the Hopfield dynamics. The corresponding discrete time dynamics are, 
\begin{align}
    x_{k+1} &=\Psi(Wx + Bu); \quad y = Cx, \label{eq:disc_firing_rate}\\
    x_{k+1} &= W\Psi(x) + Bu ; \quad y = C\Psi(x). \label{eq:disc_hopfield}
\end{align}

\begin{remark}
The state evolution in the FRNN and HNN dynamics~\eqref{eq:cts_firing_rate}-\eqref{eq:disc_hopfield} is structurally similar to a Lur'e model, which enables the use of Lemma~\ref{lem:absolute_contractivity_lure} for the analysis of contractivity. However, in classical Lur'e systems, the input enters linearly, and the output is a linear function of the state. In FRNNs, the input appears \emph{inside} the nonlinearity, and in HNNs, the output is \emph{a nonlinear function} of the state. This nonlinear coupling precludes the direct application of integral control results developed for standard Lur'e systems~\cite{MG-VA-ST-DA:23}.
\end{remark}

\begin{table*}[t]
    \centering
    \renewcommand{\arraystretch}{1.5}
    \caption{Contractivity Conditions for firing rate and Hopfield Models. See Theorem~\ref{thm:main_fr_h_conditions} for definitions of symbols. \newline The neural network with synaptic matrix $W$ corresponding to each entry is contracting (with rate $c$ or factor $\rho$) if there exist $P\succ0$ and diagonal $Q\succ0$ satisfying the corresponding LMI.}
    \setlength{\tabcolsep}{8pt}
    \begin{tabular}{@{}ll | c l r | c l r @{}}
        \toprule
        \textbf{Architecture} & \textbf{Nonlinearity} & \multicolumn{3}{c|}{\textbf{Discrete Time}} & \multicolumn{3}{c@{}}{\textbf{Continuous Time}} \\
        \midrule
        \multirow{2}{*}{\textbf{Firing Rate}} 
        & \textbf{CONE} $[-1, 1]$ & 
        $\begin{bmatrix}
        - \rho^2 P + W^\top Q W  & \0_{n \times n} \\
        \0_{n \times n} & P - Q
        \end{bmatrix}$ & $\preceq 0$ & (\refstepcounter{equation}\theequation)\label{eq:fr_disc_cone} & 
        $\begin{bmatrix}
        -2(1 - c)P + W^\top Q W & P  \\
        P  & - Q
        \end{bmatrix}$ & $\preceq 0$ & (\refstepcounter{equation}\theequation)\label{eq:fr_cts_cone} \\
        \addlinespace
        & \textbf{MONE} $[0, 1]$ &  
        $\begin{bmatrix}
        - \rho^2 P & W^\top Q \\
        Q W & P - 2Q
        \end{bmatrix}$ & $\preceq 0$ & (\refstepcounter{equation}\theequation)\label{eq:fr_disc_mone} &  
        $\begin{bmatrix}
        -2(1 - c)P & P + W^\top Q \\
        P + QW & -2Q
        \end{bmatrix}$ & $\preceq 0$ & (\refstepcounter{equation}\theequation)\label{eq:fr_cts_mone} \\
        \midrule
        \multirow{2}{*}{\textbf{Hopfield}} 
        & \textbf{CONE} $[-1, 1]$ & 
        $\begin{bmatrix}
        - \rho^2 P + Q & \0_{n \times n} \\
        \0_{n \times n} & W^\top P W - Q
        \end{bmatrix}$ & $\preceq 0$ & (\refstepcounter{equation}\theequation)\label{eq:h_disc_cone} & 
        $\begin{bmatrix}
        -2(1 - c)P + Q & PW  \\
        W^\top P  & -Q
        \end{bmatrix}$ & $\preceq 0$ & (\refstepcounter{equation}\theequation)\label{eq:h_cts_cone} \\
        \addlinespace
        & \textbf{MONE} $[0, 1]$ &              
        $\begin{bmatrix}
        - \rho^2 P & Q \\
        Q & W^\top P W - 2Q
        \end{bmatrix}$ & $\preceq 0$ & (\refstepcounter{equation}\theequation)\label{eq:h_disc_mone} & 
        $\begin{bmatrix}
        -2(1 - c)P & PW +  Q \\
        W^\top P + Q & -2Q
        \end{bmatrix}$ & $\preceq 0$ & (\refstepcounter{equation}\theequation)\label{eq:h_cts_mone} \\
        \bottomrule
    \end{tabular}
    \label{tab:lmi_conditions}
\end{table*}

\begin{thm}[Contractivity of FRNN and HNN]\label{thm:main_fr_h_conditions}
    Consider the FRNN and HNN dynamics~\eqref{eq:cts_firing_rate}-\eqref{eq:disc_hopfield}, for some constant $u$. Let $W$ be the synaptic weight matrix and $\Psi(\cdot)$ be the activation function belonging to either the CONE or MONE class. Given $P\succ0$, the dynamics are strongly contracting  with rate $c > 0$ (in continuous time) or factor $\rho \in [0,1)$ (in discrete time) in the norm $\norm{\cdot}{P}$ if there exists a diagonal  $Q\succ0$ such that the corresponding matrix inequality in Table~\ref{tab:lmi_conditions} is satisfied. 
\end{thm}

\begin{proof}
     The proof of each of the eight cases follows from  Lemma~\ref{lem:absolute_contractivity_lure} and Corollary~\ref{cor:CONE_MONE_IQC}. For simplicity, we provide the proof for only the FRNN case in continuous time for MONE nonlinearities; the other cases follow similarly. We begin by noting that the dynamics~\eqref{eq:cts_firing_rate} is a continuous-time Lur'e system of the form~\eqref{eq:cts_lure}, with $A = - I_n, B = I_n$ and $C = W$. Per Corollary~\ref{cor:CONE_MONE_IQC} a MONE nonlinearity admits the matrix multiplier $M = \begin{bmatrix}
         \0_{n \times n} & Q  \\ Q & -2Q
     \end{bmatrix}$, for some positive diagonal $Q$. Substituting these equalities into~\eqref{eq:cts_lure_lmi} yields~\eqref{eq:fr_cts_mone}.
\end{proof}

\subsection{Structural relationships of contractivity conditions}
First we show that the discrete-time CONE inequalities for both network architectures are reducible to Schur diagonal stability, generalizing prior results on firing-rate models~\cite{WDA-ALB-MF:24}.

\begin{lemma}[Schur diagonal stability]
\label{lem:discrete_cone_schur}
The discrete-time CONE conditions~\eqref{eq:fr_disc_cone} and~\eqref{eq:h_disc_cone} hold for a given $W$ if and only if $W$ is Schur diagonally stable.
\end{lemma}
\begin{proof}
    Note that ~\eqref{eq:fr_disc_cone} $\iff P \preceq Q$ and $W^\top QW \preceq \rho^2 P$. Similarly,~\eqref{eq:h_disc_cone} $\iff Q \preceq \rho^2P$ and $W^\top PW \preceq Q$.
    First, for~\eqref{eq:fr_disc_cone}, $\implies$ follows from observing that $W^\top QW \preceq \rho^2 P \preceq \rho^2 Q$. The $\impliedby$ direction follows from choosing $P = Q$. Next, for~\eqref{eq:h_disc_cone}, for the $\implies$ direction, consider the inequality $Q \preceq \rho^2 P$, and right and left multiply by $W$ and its transpose. We get $W^\top Q W \preceq \rho^2 W^\top P W \prec Q$. For the $\impliedby$ direction, choosing $\rho^2 P = Q$ gives us our result. 
\end{proof}

Next, we establish the relationships between the various conditions presented in Table~\ref{tab:lmi_conditions}.

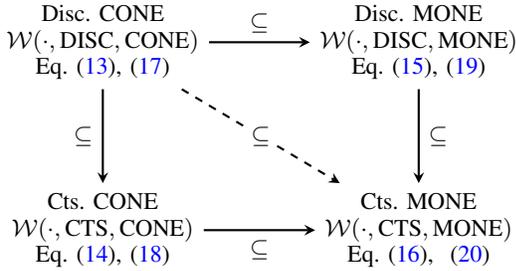
\begin{figure}[htpb]
    \centering

\begin{tikzpicture}[
    >=stealth,
    thick,
    every node/.style={align=center, font=\small}
]
    \node (DC) at (0, 2.5) {Disc. CONE \\ $\mathcal{W}(\cdot, \text{DISC}, \text{CONE})$ \\ Eq.~\eqref{eq:fr_disc_cone},~\eqref{eq:h_disc_cone}};
    \node (DM) at (4.2, 2.5) {Disc. MONE \\ $\mathcal{W}(\cdot, \text{DISC}, \text{MONE})$ \\ Eq.~\eqref{eq:fr_disc_mone},~\eqref{eq:h_disc_mone}};
    \node (CC) at (0, 0) {Cts. CONE \\ $\mathcal{W}(\cdot, \text{CTS}, \text{CONE})$ \\ Eq.~\eqref{eq:fr_cts_cone},~\eqref{eq:h_cts_cone}};
    \node (CM) at (4.2, 0) {Cts. MONE \\ $\mathcal{W}(\cdot, \text{CTS}, \text{MONE})$ \\ Eq.~\eqref{eq:fr_cts_mone}, ~\eqref{eq:h_cts_mone}};

    \draw[->] (DC) -- node[above] {$\subseteq$} (DM);
    \draw[->] (DC) -- node[left] {$\subseteq$} (CC);
    \draw[->] (CC) -- node[below] {$\subseteq$} (CM);
    \draw[->] (DM) -- node[right] {$\subseteq$} (CM);
    
    \draw[->, dashed] (DC) -- node[fill=white, inner sep=1pt] {$\subseteq$} (CM);
\end{tikzpicture}
\caption{A summary of relationships for the contractivity conditions from Table~\ref{tab:lmi_conditions}. The sets $\mcW(\cdot, \cdot, \cdot)$ are described in Theorem~\ref{thm:structural_relationships}. The discrete time CONE condition restricts the weight matrices the most, whereas the continuous time MONE condition enables maximum expressivity.}
\end{figure}

\begin{thm}[Reductions and duality of matrix inequalities]\label{thm:structural_relationships}
Let $\mcW(\mcM, \mcT, \mcN)$ denote the set of synaptic matrices~$W$ satisfying the matrix inequality in Table~\ref{tab:lmi_conditions} for model $\mathcal{M} \in \{\text{FR}, \text{Hop}\}$, time domain $\mathcal{T} \in \{\text{CTS}, \text{DISC}\}$, and nonlinearity class $\mathcal{N} \in \{\text{CONE}, \text{MONE}\}$. The following hold:
    \begin{enumerate}
        \item The set of synaptic matrices satisfying the condition for CONE nonlinearities is contained within the corresponding set for MONE nonlinearities.\label{item:nonlinearity_inclusion}
        \begin{equation*}
            \mathcal{W}(\cdot, \cdot, \text{CONE}) \subseteq \mathcal{W}(\cdot, \cdot, \text{MONE}).
        \end{equation*}
        
        \item \label{item:discretization} If a synaptic matrix $W$ satisfies any discrete-time condition with factor $\rho \in [0, 1)$, then it also satisfies the corresponding continuous-time condition with rate $c = (1 - \rho^2)/2$. Thus:
        \begin{equation*}
             \mathcal{W}(\cdot, \text{DISC}, \cdot) \subseteq \mathcal{W}(\cdot, \text{CTS}, \cdot).
        \end{equation*}
        
        \item \label{item:duality}
        A matrix $W$ satisfies the firing rate condition for some parameters $P \succ 0$ and diagonal $Q \succ 0$ if and only if $W^\top$ satisfies the corresponding Hopfield LMI condition for some parameters $P' \succ 0$ and diagonal $Q' \succ 0$.
        \begin{equation*}
             W \in \mathcal{W}(\text{FR}, \cdot, \cdot) \iff W^\top \in \mathcal{W}(\text{Hop}, \cdot, \cdot).
        \end{equation*}
    \end{enumerate}
\end{thm}

\begin{proof}
    Regarding~\ref{item:nonlinearity_inclusion}, consider the negative semidefinite matrices $M_1 = \begin{bmatrix}
        - W^\top Q W & W^\top Q \\ Q W & - Q
    \end{bmatrix}$ and $M_2 = \begin{bmatrix}
        - Q & Q \\ Q & - Q
    \end{bmatrix}$. Adding the inequality $M_1 \preceq 0$ to each of the firing rate CONE inequalities inequalities yields the corresponding firing rate MONE inequality. Similarly, adding the inequality $M_2 \preceq 0$ to each of the Hopfield CONE inequalities yields the corresponding MONE inequality.
    Next, regarding~\ref{item:discretization}, we present a proof for the firing rate MONE case. The other three cases follow similarly. Let $W$ satisfy the discrete-time condition Eq.~\eqref{eq:fr_disc_mone} for some $\rho$. Using the inverse Schur complement of the identity $P - PP^{-1}P \succeq 0$, and setting $c = (1-\rho^2)/2$, and adding and subtracting the correct terms, 
     \begin{align*}
        &\begin{bmatrix} P & -P \\ -P & P \end{bmatrix} \succeq 0 \implies \nonumber \\
        &\begin{bmatrix} -(1 + \rho^2)P & P + W^\top Q \\ P + QW & -2Q \end{bmatrix} \preceq 
        \begin{bmatrix} -\rho^2P & W^\top Q \\ QW & P -2Q \end{bmatrix}.
    \end{align*}
    Consequently, if the RHS is negative semidefinite (i.e.~\eqref{eq:fr_disc_mone} holds), so is the LHS (i.e.~\eqref{eq:fr_cts_mone} holds). Finally, for~\ref{item:duality}, we shall prove each case separately. First, in the continuous time MONE case, consider the LHS of~\eqref{eq:fr_cts_mone}. Let $T = \diag{P^{-1}, Q^{-1}}$.  Pre- and post-multiplying~\eqref{eq:fr_cts_mone} by $T$ yields:
    \begin{equation*}
         T \begin{bmatrix} -2(1 - c)P & \star \\ P + QW & -2Q \end{bmatrix} T = 
         \begin{bmatrix} -2(1 - c)P^{-1} & \star \\ W P^{-1} + Q^{-1} & -2Q^{-1} \end{bmatrix}
    \end{equation*}
    where $\star$ denotes symmetric terms. This results in the LHS of the Hopfield matrix inequality~\eqref{eq:h_cts_mone} with parameters $W' = W^\top$, $P' = P^{-1}$, and $Q' = Q^{-1}$. Next, for the continuous time CONE case, consider the Schur complement of the firing rate matrix inequality~\eqref{eq:fr_cts_cone} about the (2,2) block. We obtain, 
    \begin{align}
        -2(1 - c)P + W^\top Q W + PQ^{-1}P \preceq 0.
    \end{align}
    Upon left and right multiplying this inequality by $P^{-1}$,
    \begin{align}\label{eq:cone_duality_step}
        -2(1 - c)P^{-1} + P^{-1}W^\top Q WP^{-1} + Q^{-1} \preceq 0.
    \end{align}
    Note that~\eqref{eq:cone_duality_step} is the Schur complement of the Hopfield matrix inequality~\eqref{eq:h_cts_cone} with parameters $W' = W^\top$, $P' = P^{-1}$, and $Q' = Q^{-1}$.  Next, in the discrete time MONE case, consider the Schur complement of the firing rate matrix~\eqref{eq:fr_disc_mone} about the (1,1) block. 
    \begin{align}
        P - 2Q + \rho^{-2} QWP^{-1}W^\top Q \preceq 0.
    \end{align}
    Upon left and right multiplying this inequality by $Q^{-1}$,
    \begin{align}\label{eq:mone_duality_step}
        Q^{-1}PQ^{-1} - 2Q^{-1} + \rho^{-2} WP^{-1}W^\top  \preceq 0.
    \end{align}
    Note that~\eqref{eq:mone_duality_step} is the Schur complement of the Hopfield matrix inequality~\eqref{eq:h_disc_mone} with parameters $W' = W^\top$, $P' = \rho^{-2}P^{-1}$, and $Q' = Q^{-1}$. Finally, for the discrete time CONE case, both conditions are reducible to Schur diagonal stability as shown in Lemma~\ref{lem:discrete_cone_schur}. Therefore, we only need to show that 
    $W^\top QW \preceq \rho^2 Q$ iff $WQ^{-1}W^\top \preceq \rho^2 Q^{-1}$. Via Schur complement, observe that,
    \begin{align*}
        W^\top QW \preceq \rho^2 Q &\iff \begin{bmatrix} \rho^2 Q & W^\top Q \\ Q W & Q \end{bmatrix} \succeq 0 \\
    & \iff\begin{bmatrix} \rho^2 Q^{-1} & Q^{-1} W^\top \\ W Q^{-1} & Q^{-1} \end{bmatrix} \succeq 0. \\
    &\iff  Q^{-1} -  \rho^{-2} W Q^{-1} W^\top \succeq 0.
\end{align*}
Where the second step follows by pre and post multiplying the Schur complement by $\diag{Q^{-1},\; Q^{-1}}$, and the final step follows by applying Schur complement with respect to the (1,1) block.
\end{proof}

\begin{cor}
    For a fixed $c > 0$ or $\rho \in [0,1)$, each contractivity condition in Table~\ref{tab:lmi_conditions} is equivalent to a Linear Matrix Inequality (LMI) under a bijective change of variables. Consequently, the set of contracting synaptic matrices W, along with the Euclidean weighting P and diagonal multiplier Q, is characterized by a convex feasibility problem. 
\end{cor}

\begin{proof}
Due to the duality result between Hopfield and firing rate neural networks established in Theorem~\ref{thm:structural_relationships}, if a condition can be parameterized as an LMI for one architecture, it is inherently an LMI for the other architecture.  

For the continuous-time Hopfield networks (both CONE and MONE), applying the variable substitution $S = W^\top P$ renders the respective matrix inequalities linear in the variables $(P, Q, S)$. By duality, the continuous-time firing rate conditions are equivalently LMIs in the variables $(P^{-1}, Q^{-1}, W^\top P^{-1})$. 

In discrete time, the firing rate MONE inequality is inherently linear in the variables $(P, Q, S)$ under the substitution $S = QW$. The discrete-time Hopfield MONE case is consequently an LMI by duality. 

Finally, the discrete-time CONE conditions for both architectures reduce to Schur diagonal stability. This is well known to be an LMI; by applying the Schur complement to $W^\top Q W \preceq \rho^2 Q$ and substituting $S = QW$, we obtain the equivalent LMI $\begin{bmatrix} \rho^2 Q & S^\top \\ S & Q \end{bmatrix} \succeq 0$.
\end{proof}

\subsection{Necessary versus sufficient conditions for contractivity}

Theorem~\ref{thm:structural_relationships} implies that continuous-time FRNNs with MONE nonlinearities admit the largest set of synaptic matrices. Consequently, we focus on these dynamics with contractivity inequality~\eqref{eq:fr_cts_mone} for the rest of this letter. 
 Corresponding results for the Hopfield case are obtained via transposition.  First, we show that Lyapunov Diagonal Stability (LDS) of $W{-}I_n$ is necessary for the matrix inequality~\eqref{eq:fr_cts_mone} to hold, but not sufficient. While the insufficiency result is known~\cite[Theorem 7]{LK-ME-JJES:22},  we provide a brief proof for completeness.

\begin{pro}[Relationship with Lyapunov Diagonal Stability]\label{prop:LDS}
    Let $W \in \R^{n\times n}$ be a synaptic matrix for the firing rate neural network~\eqref{eq:cts_firing_rate}. The following hold
    \begin{enumerate}
        \item If $W$ satisfies~\eqref{eq:fr_cts_mone} for $P \succ 0$, diagonal $Q \succ 0$ and $c > 0$, then $W - I_n$ is $Q-$LDS. \label{item:lds_i}
        \item The converse is not true; $W - I_n$ being LDS is not sufficient to guarantee contraction in any fixed norm. \label{item:lds_ii}
    \end{enumerate}
\end{pro}
\begin{proof}
    To prove~\ref{item:lds_i}, we left and right multiply the matrix inequality~\eqref{eq:fr_cts_mone} by the matrix $[I_n, I_n] \in \R^{n \times 2n}$ obtaining,
    \begin{align*}
    -2(1-c) P + P + W^\top Q + P + QW - 2Q &\preceq 0 \\
    \implies W^\top Q + QW  \preceq 2Q - 2cP &\prec 2Q.
    \end{align*}
    This proves that the $W - I_n$ is $Q-$LDS. 
    
    We construct a counterexample for~\ref{item:lds_ii}. Consider a skew-synaptic weight matrix $ W = \begin{bmatrix} 0 & 4 \\ -4 & 0 \end{bmatrix}$. Since $W + W^\top  = 0 \prec 2 I$, $W-I_n$  is LDS. For contraction in the norm $\norm{\cdot}{P}$, the matrix $M(D) = 2P - PDW - W^\top DP $ must be positive definite for all diagonal $D \in [0,I]$, e.g., see~\cite{VC-AG-AD-GR-FB:23c}.
    Testing the vertices $D_1 = \diag{1, 0}$ and $D_2 = \diag{0, 1}$,
    \begin{align*}
        M(D_1) &= \begin{bmatrix}
            2 p_{11} & 2p_{12} - 4 p_{11} \\
            2p_{12} - 4 p_{11} & 2p_{22} - 8p_{12} 
        \end{bmatrix},  \\
        M(D_2) &= \begin{bmatrix}
            2p_{11} - 8 p_{12} & 2p_{12} - 4 p_{22} \\
            2p_{12} - 4 p_{22} & 2 p_{22}
        \end{bmatrix} . 
    \end{align*}
    Since the determinants of these matrices need to be positive, 
    \begin{itemize}
        \item $p_{11} p_{22} > p_{12}^2 + 4p_{11}^2 \implies p_{22} > 4p_{11}$.
        \item $p_{11} p_{22} > p_{12}^2 + 4p_{22}^2 \implies p_{11} > 4p_{22}$.
    \end{itemize}
    These conditions are contradictory for any $P \succ 0$. Thus, $W-I_n$ is LDS 
    but not contractive.        
\end{proof}

Next, we show that restricting the matrix inequality~\eqref{eq:fr_cts_mone} to
symmetric synaptic matrices is equivalent to the known
optimal characterizations given in~\cite{VC-AG-AD-GR-FB:23c}.

\begin{pro}[Log-optimal characterization for symmetric weights]\label{prop:symm_matrix}
    Let $W = W^\top \in \R^{n \times n}$. 
    \begin{enumerate}
        \item For $\alpha(W) < 0$, the matrix inequality~\eqref{eq:fr_cts_mone} holds for ${P = - W}$, ${Q = I_n}$, and ${c = 1}$. 
        \item For $0 < \alpha(W) < 1$, there exists $P \succ 0$ such that $ W = P ^{1/2} - (4\alpha(W))^{-1}P$, and the matrix inequality~\eqref{eq:fr_cts_mone} holds for this  $P$, ${Q = 4\alpha(W)I_n}$, and $c = 1 - \alpha(W)$. 
    \end{enumerate}
\end{pro}
\begin{proof}
    Consider the Schur complement of~\eqref{eq:fr_cts_mone}. 
    \begin{align}
        -2(1 - c)P + \frac{1}{2}(P + W^\top Q)Q^{-1}(P + QW) \preceq 0.\label{eq:schur_fr_cts_mone}
    \end{align}
    The first statement follows from substituting $P = - W$, $Q = I_n$ and $c = 1$ into~\eqref{eq:schur_fr_cts_mone}. \\
From~\cite[Lemma 10]{VC-AG-AD-GR-FB:23c}, we may define $P \succ 0$ such that ${W = P^{1/2} - \frac{1}{4\alpha(W)}P}$. Choosing $Q = 4\alpha(W) I$, and substituting $W$ and $Q$ into the LHS of~\eqref{eq:schur_fr_cts_mone} yields ${-2(1 - c)P + 2\alpha(W)P}$, which is $0$ when $c = 1 - \alpha(W)$.
\end{proof}

\section{Applications to integral control and machine learning} \label{sec:applications}
\subsection{Parameterization}
In this section we study the parameterization of synaptic matrices satisfying the inequality~\eqref{eq:fr_cts_mone}; such a step enables applications in machine learning and data-driven control.
\begin{thm}[Parameterization of weight matrices]\label{thm:parameterization}
    The following  statements are equivalent:
    \begin{enumerate}
        \item The synaptic matrix $W$ satisfies condition~\eqref{eq:fr_cts_mone} for some $P \succ 0$, diagonal matrix $Q\succ 0$, and $c \in [0,1]$.
        \item $W$ can be written as
            \begin{align*}          
                W = 2\sqrt{1 {-} c} \; \diag{\e^d} S V^\top V  - \diag{\e^{2d}}(V^\top V )^2 
            \end{align*}
            where $S \in \R^{n\times n}$ satisfying $S^\top S \preceq I_n$, $V$ is a full rank matrix in $\R^{n\times n}$, and $d\in \R^n$. 
    \end{enumerate}
\end{thm}
\begin{proof}
    Consider the (1,1) Schur complement of the inequality~\eqref{eq:fr_cts_mone}. With some bookkeeping, we obtain,
    \begin{align}\label{eq:begin_param}
        (P + W^\top Q)Q^{-1} (P + QW) \prec 4(1-c)P.
    \end{align}
    From the Douglas-Fillmore-Williams Lemma~\cite[Theorem 8.6.2]{DSB:09}, the inequality~\eqref{eq:begin_param} is true if and only if there exists a matrix $S \in \R^{n\times n}$ satisfying $S^\top S \preceq I_n$ and
    \begin{align}
        Q^{-1/2} (P + QW) = 2\sqrt{(1-c)}\;SP^{1/2}.
    \end{align}
    Upon rearranging to solve for $W$, 
     \begin{align}
        W = 2\sqrt{1 - c} \; Q^{-1/2} S P^{1/2} - Q^{-1}P.
    \end{align}
    Now, the set of all valid $Q$ can be parameterized using $\diag{\e^{-2d}}$, and the
    set of all $P \succ 0$ can be parameterized as $V^\top V$ for full rank $V \in \R^{n\times n}$.
\end{proof}
\begin{remark}
In numerical applications, the set of all matrices $S$ such that $S^\top S \preceq I_n$ may be parameterized via a free variable $X \in \R^{n\times n}$ by setting $S=X(I+X^{\top}X)^{-1/2}$. Similarly, to ensure $V^\top V$ does not lose rank, we may approximate $V^\top V$ using a free matrix $Y \in \R^{n \times n}$ and a small constant $\epsilon > 0$ by setting $P = Y^\top Y + \epsilon I_n$.
\end{remark}

Theorem~\ref{thm:parameterization} provides a practical tool to design FRNNs that are contracting by construction. Having characterized the largest known set of synaptic matrices that guarantee contractivity for the continuous-time model~\eqref{eq:cts_firing_rate}, we now turn to utilizing these models in control and machine learning.

\subsection{LMI conditions for integral control} 

Theorem~\ref{thm:parameterization} enables the modeling of stable continuous-time systems using contracting FRNNs. Given such a well-behaved, data-driven model, a natural next step is to achieve robust reference tracking. In this section, we establish conditions that enable the use of low-gain integral control, deriving these conditions via a singular perturbation argument~\cite{JWSP:21}.

We treat the FRNN as fast dynamics, and the integrating controller as slow dynamics. For a constant reference $r \in \R^p$, consider the dynamics,
\begin{subequations}\label{eq:closed_loop_system}
\begin{align}
    \dot{x} &= -x + \Psi\left(Wx + Bu\right), \quad y = Cx \label{eq:contracting_FRNN} \\
    \dot{u} &= \varepsilon K(r - y) \label{eq:integral_controller}
\end{align}
\end{subequations}
where $x \in \R^{n}, u \in \R^m, y \in \R^p$, and all matrices are shaped appropriately. When~\eqref{eq:contracting_FRNN} is contracting, the fixed point map $\xstar(u)$ is well posed, and is implicitly defined by
\begin{align}
\xstar(u) &= \Psi\left( W \xstar(u) + B u\right). \label{eq:cl_fixed_point_equation} 
\end{align}
Further, the reduced order dynamics are given by
\begin{align} \label{eq:reduced_order_dynamics}
    \dot{u} &= K (r - C\xstar(u)).
\end{align}

Before presenting our main result, we introduce an incremental matrix multiplier for this fixed point map from~\cite{LDA-MC:13}.
\begin{lemma}\label{lem:implicit_imm}
    Let the map $\map{\Psi}{\R^n}{\R^n}$ admit an incremental matrix multiplier $M$. Assume the implicit definition of the map $\xstar(u)$ in equation~\eqref{eq:cl_fixed_point_equation} is well posed. Then the map $\xstar(u)$ admits the incremental  multiplier matrix
    \begin{align}
        \begin{bmatrix}
            B & W \\ 0 & I_n 
        \end{bmatrix}^\top
        M
        \begin{bmatrix}
            B & W \\ 0 & I_n 
        \end{bmatrix} 
    \end{align}
\end{lemma}

Next, we derive conditions guaranteeing that the reduced-order dynamics~\eqref{eq:reduced_order_dynamics} are contracting. Combined with the contractivity of the open-loop plant, these conditions ensure that the closed-loop system achieves asymptotic reference tracking via the low-gain integral control analysis of~\cite{JWSP:21}.

\begin{thm}[Contractivity of the reduced dynamics and reference tracking via integral control]\label{thm:dc_gain}
Consider the system~\eqref{eq:closed_loop_system}. Let $A = I_n - W$, and $Q\in \R^{n\times n}$ be a diagonal positive matrix. Assume:
\begin{enumerate}[label=\textup{(A\arabic*)},leftmargin=*]
    \item $\Psi$ is slope-restricted to $[\delta, 1]$.
    \item  There exists $\ustar$ such that $r = C\xstar(\ustar)$, and $\xstar(\ustar) = \Psi(W\xstar(\ustar) + B\ustar)$.
    \item With the shorthands ${Z = B^\top Q ((1 - \delta)I_n  + 2\delta A)}$ and ${R = 2\delta A^\top Q A + (1 - \delta) (QA + A^\top Q)}$, there exist matrices $Y$ and $P\succ 0$ such that
    \begin{align}\label{eq:dc_gain_contraction_lmi}
     \begin{bmatrix}
        2c_rP - 2 \delta B^\top QB &  Z -YC \\    * & -R
      \end{bmatrix}  \preceq 0.
    \end{align}
\end{enumerate}
Then, with $K = P^{-1}Y$, 
\begin{enumerate}
    \item the reduced order dynamics~\eqref{eq:reduced_order_dynamics} are strongly infinitesimally contracting with rate $c_r$ with respect to the $\norm{\cdot}{P}$, and
    \item if the plant dynamics~\eqref{eq:contracting_FRNN} are strongly infinitesimally contracting, then there exists $\varepsilon^* > 0$ such that for all $\varepsilon \in (0, \varepsilon^*)$, the closed-loop system~\eqref{eq:closed_loop_system} is asymptotically stable and achieves reference tracking.
    \end{enumerate}
\end{thm}
\begin{proof}
    We begin by noting that the reduced order dynamics~\eqref{eq:reduced_order_dynamics} form a Lur'e system, where the nonlinearity is given by $\xstar(u)$. We may use Lemma~\ref{lem:implicit_imm} in conjuction with Lemma~\ref{lem:absolute_contractivity_lure}, to obtain the following LMI condition for contractivity in a Euclidean norm weighted by $P$. 
    \begin{align}\label{eq:int_condition_step_1}
     \begin{bmatrix}
        2c_rP & -PKC \\    -CKP & \0_{m\times{m}}
      \end{bmatrix} + \lambda
       \begin{bmatrix}
            B & W \\ 0 & I_n 
        \end{bmatrix}^\top
        M 
        \begin{bmatrix}
            B & W \\ 0 & I_n 
        \end{bmatrix} \preceq 0.
    \end{align}
    where $M$ is the matrix multiplier for the nonlinearity $\Psi$. Using Lemma~\ref{lem:IMM_conditions}, we may substitute $M$ as $\begin{bmatrix}
            -2\delta Q &  (1 + \delta) Q \\
             (1 + \delta) Q & -2Q
        \end{bmatrix}$. The second term in~\eqref{eq:int_condition_step_1} resolves to
    \begin{align*}
        \begin{bmatrix}
         - 2 \delta B^\top QB &  - 2\delta B^\top QW + (1 + \delta) B^\top Q  \\    * &  (1 + \delta) (W^\top Q + QW) - 2Q - 2\delta W^\top QW 
      \end{bmatrix}
    \end{align*}
    Now, let us set $A = I_n - W$. First, note that 
    \begin{align*}
        - R &=- 2\delta W^\top QW + (1 + \delta) (W^\top Q + QW) - 2Q,  \\
        Z &= - 2\delta B^\top QW + (1 + \delta) B^\top Q.
    \end{align*}
    Now, we may rewrite~\eqref{eq:int_condition_step_1} using these simplifications. 
    \begin{align}\label{eq:dc_gain_contraction_lmi_prep}
     \begin{bmatrix}
        2c_rP - 2 \delta B^\top QB &  Z -PKC \\    * & -R
      \end{bmatrix}  \preceq 0.
    \end{align}
    Reparameterizing $Y = PK$ we get the condition~\eqref{eq:dc_gain_contraction_lmi}.

Regarding the second statement, when the plant and the reduced order dynamics are globally strongly contracting, all of the assumptions in~\cite[Theorem 3.1]{JWSP:21} are met and therefore, for sufficiently small $\varepsilon$, the closed-loop system is asymptotically stable. 
\end{proof}    

\begin{cor}\label{cor:linear_dc_gain}
Under the assumptions of Theorem \ref{thm:dc_gain}, 
\begin{align}
    P (KCA^{-1}B) + (KCA^{-1}B)^\top P \succ 2c_r P, 
\end{align}
i.e., the matrix $- KCA^{-1}B$ is Hurwitz, recovering the standard DC-gain condition for the linear case~\cite{JWSP:21}.
\end{cor}

\begin{proof}
Consider the condition \eqref{eq:dc_gain_contraction_lmi}. Pre and post multiplying both sides by $\diag{I_m, (A^{-1}B)^\top}$ and $\diag{I_m, (A^{-1}B)}$ respectively, we obtain, 
\begin{align}\label{eq:dc_gain_contraction_lmi_transformed}
     \begin{bmatrix}
        2c_rP - 2 \delta B^\top QB &  ZA^{-1}B -YCA^{-1}B \\   (A^{-1}B)^\top (Z -YC)^\top & -(A^{-1}B)^\top RA^{-1}B
      \end{bmatrix}  \preceq 0.
\end{align}
We evaluate each term. First, we note the following, 
\begin{align}
    (A^{-1}B)^\top RA^{-1}B &= (1 - \delta) B^\top ((A^\top)^{-1}Q +  QA^{-1})B \nonumber\\    
    & + 2\delta B^\top Q B  \\
    ZA^{-1}B &= B^\top Q ((1 - \delta)A^{-1}  + 2\delta I) B
\end{align}
Upon making these substitutions, we further pre and post multiply~\eqref{eq:dc_gain_contraction_lmi_transformed} by the block rectangular matrix, $T_2 = [I_m, I_m]$, and its transpose. Upon simplifying, we obtain, 
\begin{align}
    2c_rP  
    -YCA^{-1}B - (YCA^{-1}B)^\top
    \preceq 0.
\end{align}
Upon substituting $Y = PK$, the result follows. 
\end{proof}

For numerical validation, we apply our controller synthesis on a normalized version of the two tank benchmark system~\cite{MS-JPN:17}. We utilize the parameterization in Theorem~\ref{thm:parameterization} to perform system identification. We fit the plant to an FRNN with $n = 8$ neurons and $\tanh{(\cdot)}$ activations to enforce nonzero $\delta$ condition. We solve for the DC gain $K$ using the LMI condition in Theorem~\ref{thm:dc_gain}. As shown in Figure~\ref{fig:data_driven_control}, our controller successfully performs reference tracking. 
\begin{figure}
    \centering
    \includegraphics[width=0.97\linewidth]{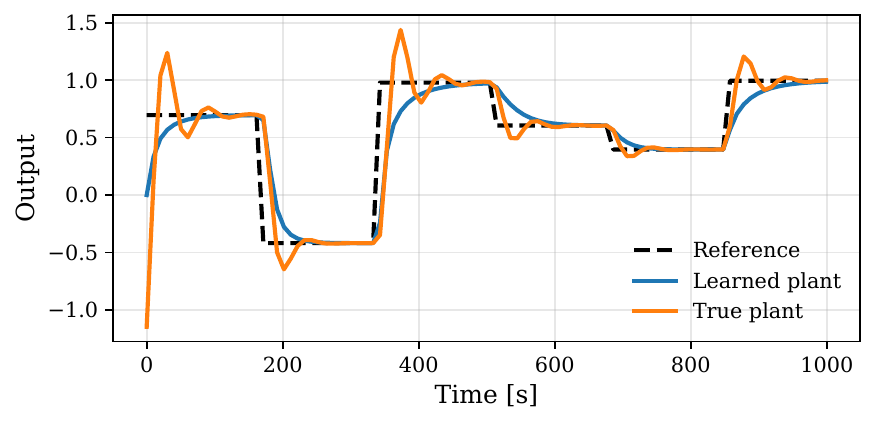}
    \caption{We perform system identification on a standard two-tank system, using an FRNN parameterized via Theorem~\ref{thm:parameterization}. We design the gain via a low-gain controller as described in Theorem~\ref{thm:dc_gain}.}
    \label{fig:data_driven_control}
\end{figure}

\subsection{Applications to implicit neural networks} \label{sec:INN}

Recent literature~\cite{JL-LD-SO-WY:25} establishes that the representational capacity of implicit models is fundamentally enhanced when the equilibrium mapping is allowed to exhibit local, rather than global, Lipschitz continuity. Motivated by this insight, we design an implicit neural network with input-dependent synaptic weights and biases:
\begin{align}
    \xstar(u) = \Psi(W(u)\xstar(u) + B(u)).
\end{align}
By allowing $W$ and $B$ to depend on the input $u$, the equilibrium $\xstar(u)$ becomes locally Lipschitz in $u$ (instead of globally Lipschitz), enabling the model to adapt its internal dynamics to each input and thereby improving expressivity. Concretely, we parameterize $W$ according to Theorem~\ref{thm:parameterization}, replacing each free variable with a learned neural network, and similarly implement $B$ as a neural network. This input-dependent network architecture natively allows $\xstar(u)$ to be locally Lipschitz, while remaining mathematically guaranteed to be contracting. We validate this approach on the MNIST and CIFAR-10 image classification benchmarks. As shown in Table~\ref{tab:image_classification}, our model achieves competitive accuracy while remaining parameter efficient, due to a higher expressivity. Full implementation details and code are provided in our repository.\footnote{\url{https://github.com/AnandGokhale/Contractivity_Neural_Networks}}

\begin{table}[htbp]
    \centering
    \caption{Comparative Performance on Image Classification Benchmarks}
    \label{tab:image_classification}
    \begin{tabular}{lcc}
        \toprule
        \textbf{Method} & \textbf{Model size} & \textbf{Acc.} \\
        \midrule
        \multicolumn{3}{c}{\textbf{MNIST}} \\
        \midrule
        LBEN~\cite{MR-RW-IRM:20} & -- & 98.2\% \\
        monDEQ~\cite{EW-JZK:20}  & 84K & 99.1$\pm$0.1\% \\
        \textbf{Ours} & \textbf{89K} & \textbf{99.33\%} \\
        \midrule
        \multicolumn{3}{c}{\textbf{CIFAR-10}} \\
        \midrule
        LBEN~\cite{MR-RW-IRM:20} & -- & 71.6\% \\
        monDEQ~\cite{EW-JZK:20}  & 172K & 74.0$\pm$0.1\% \\
        monDEQ$^*$~\cite{EW-JZK:20}  & 854K & 82.0$\pm$0.3\% \\
        \textbf{Ours} & \textbf{134K} & \textbf{78.27\%} \\
        \textbf{Ours}$^*$ & \textbf{134K} & \textbf{82.30\%} \\
        \bottomrule
    \end{tabular}
    
    \vspace{2pt}
    \raggedright \footnotesize $^*$ indicates models trained with data augmentation.
\end{table}

\section{Conclusion}

In this paper, we established sharp LMI conditions characterizing synaptic weights that guarantee absolute contractivity for firing-rate and Hopfield neural networks. We provide a detailed analysis of these LMI conditions, including relationships under different classes of nonlinearities. We also show the relationship of these conditions with prior conditions, showing that our conditions are tight when the synaptic matrix is symmetric. By translating our continuous-time MONE condition into an exact algebraic parameterization, we bridge the gap between stability analysis and practical synthesis. We demonstrated the utility of this parameterization through two distinct applications: a novel LMI-based design for low-gain integral controllers capable of robust reference tracking, and an input-dependent implicit neural network architecture that safely increases the expressivity of implicit models, by allowing the output of an implicit layer to be locally Lipschitz in the input.

Future research will proceed along several avenues. First, an interesting problem is to extend our integral control gain design to more general classes of neural networks such as RENs~\cite{DM-CLG-IRM-LF-GFT:23}. Second, one could also study mechanisms to stabilize FRNNs using output feedback~\cite{MG-MJ-ST:24}. Finally, we seek to extend our analysis to network problems, both in the context of graph neural networks, and in the context of distributed control.

\bibliographystyle{plainurl+isbn}
\bibliography{alias, Main, FB}

\end{document}